\documentclass[11pt,a4,reqno,tbtags]{amsart}

\usepackage{amsmath}
\usepackage{amsfonts}
\usepackage{amssymb}
\usepackage{latexsym}
\usepackage{graphicx}
\usepackage{epsfig}
\usepackage{epstopdf}
\usepackage[utf8]{inputenc}
\usepackage{mathrsfs}
\usepackage[footnotesize, SL]{subfigure}
\usepackage[small,bf]{caption}	
\usepackage{xy}
\numberwithin{equation}{section}
\newtheorem{proposition}{Proposition}[section]

\newtheorem{theorem}[proposition]{Theorem}

\newtheorem{lemma}[proposition]{Lemma}

\newcommand{\G}[1]{\Gamma\left(#1\right)}

\theoremstyle{definition}
\newtheorem*{ack}{Acknowledgement}

\begin{document}
     
\addtolength{\baselineskip}{0.3\baselineskip}
     
\pagestyle{empty}

\hfill

     \vspace{0.5cm}
     
     \begin{center}

     {\Large \bf Markov branching in the \\ vertex splitting model }

     \medskip
     \vspace{0.8 truecm}
     
     {\large \bf \today}
     
     \vspace{0.8truecm}

{\bf Sigurdur Örn Stefánsson}

\vspace{0.4 truecm}

{NORDITA,

Roslagstullsbacken 23, SE-106 91 Stockholm,

Sweden

}

\vspace{0.4 truecm}
\texttt{
sigste@nordita.org
}
\vspace{.3 truecm}

\end{center}
  
\small \noindent {\bf Abstract.} We study a special case of the vertex splitting model which is a recent model of randomly growing trees. For any finite maximum vertex degree $D$, we find a one parameter model, with parameter $\alpha \in [0,1]$ which has a so--called Markov branching property. When $D=\infty$ we find a two parameter model with an additional parameter $\gamma \in [0,1]$ which also has this feature. In the case $D = 3$, the model bears resemblance to Ford's $\alpha$--model of phylogenetic trees and when $D=\infty$ it is similar to its generalization, the $\alpha\gamma$--model. For $\alpha = 0$, the model reduces to the well known model of preferential attachment.  

In the case $\alpha > 0$, we prove convergence of the finite volume probability measures, generated by the growth rules, to a measure on infinite trees which is concentrated on the set of trees with a single spine. We show that the annealed Hausdorff dimension with respect to the infinite volume measure is $1/\alpha$. When $\gamma = 0$ the model reduces to a model of growing caterpillar graphs in which case we prove that the Hausdorff dimension is almost surely $1/\alpha$ and that the spectral dimension is almost surely $2/(1+\alpha)$. We comment briefly on the distribution of vertex degrees and correlations between degrees of neighbouring vertices.

\normalsize
  
\newpage
\pagestyle{plain}
    
\section {Introduction}

Random trees are an important tool in many branches of science, ranging from quantum gravity models \cite{specCDT,legalllimit} to biological applications \cite{aldous,RNAfolding}, to name a few. In this paper we introduce and study a new model of randomly growing, rooted, planar trees which we refer to as the attachment and grafting model, or ag--model for short. It is a special case of the vertex splitting model, recently introduced in \cite{vs}. The vertex splitting model is a modification of a model of growing trees, encountered in the theory of random RNA folding \cite{RNAfolding}.   

The ag--model is described informally below and a more detailed description is given in Section \ref{s:rpt}.  The root of the tree is simply a marked vertex of degree one and the planarity condition means that edges are ordered around vertices. The parameters of the model are $\alpha,\gamma \in [0,1]$ and $D \in \{3,4,\ldots\}\cup \{\infty\}$ denotes the  maximum degree of vertices in the trees.  When $D < \infty$, $\alpha$ and $D$ are the only active parameters of the model but when $D = \infty$, $\gamma$ also plays a role. Define
\begin {equation} \label {aandb}
 a = \left\{ \begin {array}{cc}
 -\frac{1-\alpha}{D-2}  &  \quad \text{if } D< \infty \\
 1-\gamma  & \quad \text{if } D= \infty \\
 	\end {array}\right. , \qquad 
\end {equation}
and
\begin {equation}
 \overline{w}_k = \left\{ \begin {array}{cc}
 a k + 1 - 2a - \alpha & \qquad \text{for } k \leq D \\
 0 & \qquad \text{for } k > D.\\
 	\end {array}\right. 
\end {equation}
The growth rules can be explained as follows. Call the edges which are adjacent to vertices of degree one (besides the root) {\it leaves} and call the other edges {\it internal edges}. In each discrete time step a new edge is added by randomly selecting
\begin {list}{\labelitemi}{\leftmargin=2em}
 \item [(a)] a vertex of degree $k\geq 2$ with relative probability $\overline{w}_k$ and attaching a new edge to it (the $k$ possibilities of attaching chosen uniformly at random) or 
\item [(b)]  an inner edge with relative probability $\alpha$ and dividing it into two edges by grafting a vertex to it or 
\item[(c)] a leaf with relative probability $1-a$ and dividing it into two edges by grafting a vertex to it,
\end {list}
see Fig.~\ref{f:growth} (left). It is from these two operations 'attachment' and 'grafting' which the model gets its name.
\begin{figure} [!h]
\centerline{\scalebox{0.24}{\includegraphics{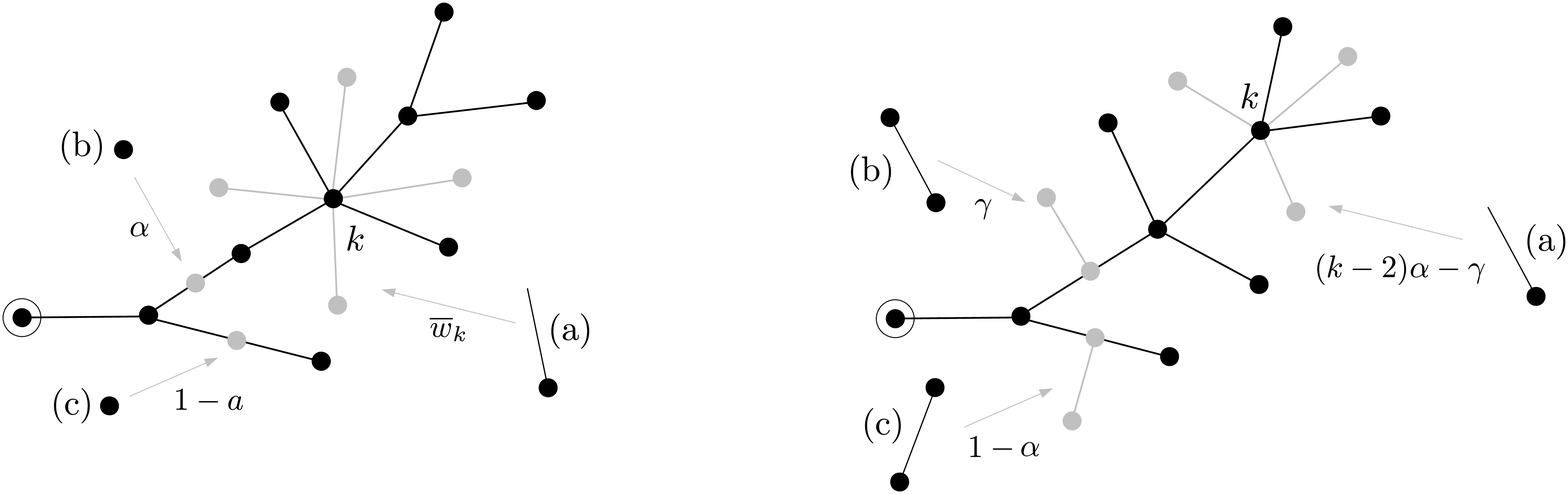}}}
\caption{Growth rules of the ag--model (left) and the $\alpha\gamma$--model (right). The root is indicated by a cirlced vertex.} \label{f:growth}
\end{figure}
When $\alpha = 0$ there is no grafting and the model reduces to the model of preferential attachment (see e.g.~\cite{albert0,massdist}) with linear attachment kernel $\overline{w}_k$. When $\alpha = 1$ there is no attaching and we simply have a growing linear graph.

The ag--model closely resembles the $\alpha\gamma$--model which was introduced in \cite{alphagamma}. In the $\alpha\gamma$--model a new leaf is added in each time step by randomly selecting 
\begin {list}{\labelitemi}{\leftmargin=2em}
 \item [(a)] a vertex of degree $k\geq 3$ with relative probability $(k-2)\alpha -\gamma$ and attaching a new edge to it or 
\item [(b)]  an inner edge with relative probability $\gamma$ and grafting a leaf to it or 
\item[(c)] a leaf with relative probability $1-\alpha$ and grafting a leaf to it
\end {list}
where $0\leq \gamma \leq \alpha \leq 1$, see Fig.~\ref{f:growth} (right). In the case $\gamma = \alpha$, this model reduces to Ford's $\alpha$--model of growing binary trees  \cite{Ford}  in which case it resembles the ag--model with $D=3$. When $\gamma<\alpha$ it is similar to the ag--model with $D=\infty$. 

The $\alpha\gamma$--model and the ag--model both have a property referred to as Markov branching which was introduced by Aldous in \cite{aldous}. This means, crudely, that the subtrees below a given vertex have the same distribution as the whole tree, see Section \ref{s:mb}. This feature makes the models much simpler to treat, since one can easily write recursion equations for many observables. Furthermore, recent results by Haas and Miermont provide a recipe for taking the scaling limit of such models \cite{markovss}.

The main results of this paper are the following. For $\alpha >0$, as the size of the trees goes to infinity, the measure concentrates on the set of trees with exactly one non-backtracking path from the root to infinity, referred to as an {\it infinite spine}.  The emergence of a unique infinite spine is known in other models of random trees, an example being the uniform planar tree and modifications of it, see e.g. ~\cite{sdgt}. Similar effects are also observed in triangulation models in quantum gravity, where exactly one large ``universe'' appears with finite baby-universes attached, see e.g.~\cite{krikun}. We also establish that the average volume of a graph ball of radius $R$ in the infinite trees grows like $R^{1/\alpha}$. The exponent is referred to as the {\it Hausdorff dimension} and denoted by $d_H$. This power law behaviour is interesting since it is often the case that models of growing trees exhibit an exponential volume growth. This is e.g.~the case in the preferential attachment model \cite{albert0} which in fact corresponds to the case $\alpha = 0$ as was noted before. Furthermore, since $\alpha \in [0,1]$, the full range of exponents, $d_H$, is realized and the ag--model is one of few known natural tree models having this feature.

\subsection{Relation to the vertex splitting model}  We now briefly introduce the vertex splitting model and show how the ag--model can be seen as a special case. The parameters of the vertex splitting model are given by a set of non--negative weights $w_{i,j}$, with $1 \leq i,j,i+j-2 \leq D$, where $D\geq 2$ (or $D=\infty$) is a fixed number which denotes the maximum vertex degree in the trees. These weights are referred to as {\it partitioning weights} and the so--called {\it splitting weights} are defined by
\begin{equation}
 w_i = \frac{i}{2} \sum_{j=1}^{i+1} w_{j,i+2-j},\qquad 1 \leq i \leq D.
\end{equation}
Starting from a fixed finite planar tree, in each discrete time step a new edge is added as follows.
\begin {list}{\labelitemi}{\leftmargin=2em}
 \item [(a)] Select a given vertex $v$ of degree $i$ with relative probability $w_i$.
\item [(b)] Randomly partition the edges which contain $v$ into two disjoint sets of adjacent edges:  $V$ of size $k-1$ and $V'$ of size $i-(k-1)$, with probability $w_{k,i+2-k}/w_i$. For a given $k$, all such partitionings are taken to be equally likely.
\item [(c)] Move all edges in $V'$ from $v$ to a new vertex $v'$ and join $v$ and $v'$ by a new edge.
\end {list}
We allow a small generalization of the above growth rule: single out a vertex of degree one in the initial tree and call it the root and modify (a) in such a way that the root is selected with relative probability $w_r \geq 0$. Each time the root is split we define the new vertex of degree one to be the root. The vertex splitting model has very general growth rules and it includes many other models of random trees as special cases or limiting cases, see \cite{vs,phdsos} for more detailed discussion.

The  ag--model can be recovered from the vertex splitting model by assigning the weight $w_r = \alpha/2$ to splitting the root and choosing the nonzero partitioning weights as follows
\begin {eqnarray} \nonumber
 w_{2,2} &=& \alpha, \\ \nonumber
w_{k,2} &=& \frac{\alpha}{2}, \quad 3\leq k \leq D, \\ \label{partw}
w_{k+1,1} &=& \frac{\overline{w}_{k}}{k}, \quad 2 \leq k \leq D.
\end {eqnarray}
The splitting weights are then
\begin {equation} \label{splittingweights}
 w_k = \left(\frac{\alpha}{2} + a\right) k + 1 - 2a - \alpha, \quad 2 \leq k \leq D
\end {equation}
and $w_k = 0$ if $k>D$. Note that in the case $D=\infty$ and $\gamma < \alpha/2$, the weight $w_1$ is negative. We will however include this case in the ag--model since the total weight of any transition is still positive. 

A similar relationship between the $\alpha\gamma$--model and the vertex splitting model was discussed in \cite{vs}. However, in that case one needs to take $w_{3,1} = \infty$ which means that comparison of results in the two models is not necessarily reliable. The ag--model is therefore more interesting as a special case and due to its simplicity, yet non--triviality, it serves as a good testing ground for non--rigorous results obtained in the vertex splitting model.

\subsection{Outline} The paper is organized as follows. In Section \ref{s:rpt} we define rooted planar trees and introduce a convenient notation for representing random trees. Thereafter we give a proper definition of the ag--model which was described informally above. In Section \ref{s:mb} we show that the model has the Markov branching property and we calculate its first split distribution. In Section \ref{s:cfvm} we show, using methods from \cite{siggi}, that the finite volume probability measures generated by the random growth operation, converge to a measure on the set of infinite trees. Furthermore, we characterize the infinite volume measure. In Section \ref{s:hd} we calculate the annealed Hausdorff dimension with respect to the infinite volume measure and in a certain special case, we calculate the almost sure Hausdorff and spectral dimensions. The results we obtain, support certain scaling assumptions which were made in the vertex splitting model.  We conclude by commenting on the distribution of the degrees of vertices in the trees and correlations between degrees of neighbouring vertices by recalling results from \cite{vs}. In order to improve readability, proofs of theorems and lemmas are in most cases collected in Appendix B. 
  
\section {Random planar trees} \label{s:rpt}
In this section we begin by defining the set of rooted, planar trees and endow it with a metric. Then we define a convenient notation for representing random trees and introduce the model which will be studied in the paper.

Start with a tree graph $\tau$ which has vertices of finite or countably infinite degree and at least one vertex of degree one. By convention we define the root $r$ of $\tau$ to be a vertex of degree one and we label the unique nearest neighbour of the root by $(1)$. The rest of the vertices are labeled in the following recursive way. The children of a given vertex in the tree (apart from $r$) with label $(\ell)$ are labeled with sequences $(\ell,1), (\ell, 2),\ldots$, see Fig.~\ref{f:planartree}. A rooted planar tree is a tree $\tau$ along with such a lexicographical labeling. From here on, we will always work with rooted, planar trees unless otherwise stated and will simply refer to them as trees. We denote the set of trees with $n$ edges by $\mathcal{T}_n$ and the set of all trees, finite and infinite, by $\mathcal{T}$.
\begin{figure} [!h]
\centerline{\scalebox{0.8}{\includegraphics{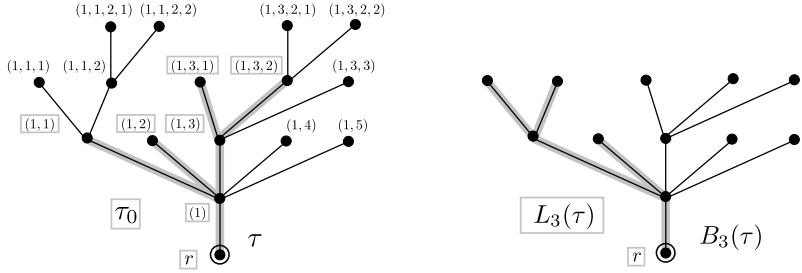}}}
\caption{Left: An example of a rooted, planar tree $\tau$ and a left subtree $\tau_0$ (boxed in gray). Right: The graph ball $B_3(\tau)$ and the left ball $L_3(\tau)$ (boxed in gray). The root is indicated by a circled vertex.} \label{f:planartree}
\end{figure}

A tree $\tau_0$ is said to be a {\it left subtree} of $\tau$ if it is a connected subtree of $\tau$ which contains $r$ and has the properties that if it contains a vertex with label $(\ell,k)$ then it contains all vertices with labels $(\ell,i)$ with $i\leq k$, see Fig.~\ref{f:planartree}. Let $B_R(\tau)$ be the graph ball of radius $R$ centered on the root of $\tau$. We define the {\it left ball} of radius $R$, $L_R(\tau)$, as the maximal left subtree of $B_R(\tau)$ with vertices of degree no greater than $R$, see Fig.~\ref{f:planartree}.  A metric $d$ is defined on $\mathcal{T}$ by
\begin {equation}
d(\tau,\tau') = \inf\left\{\frac{1}{R}~\Bigg| ~ L_R(\tau) = L_R(\tau')\right\}.
\end {equation}
The metric $d$ was first introduced in \cite{ngtrees} and we refer to this paper for some properties of the metric space $(\mathcal{T},d)$.

Define the root joining operation $\ast$ in the following way. Given trees $\tau_1, \tau_2, \ldots, \tau_{k}$, $k\geq 1$, let $\tau = \emptyset \ast \tau_1\ast \tau_2 \ast \cdots \ast \tau_k$ be the tree obtained by (I) identifying the roots of $\tau_1, \tau_2,\ldots,\tau_k$ and labeling them by (1), (II) replacing the first element '$1$' of each label in $\tau_j$ by '$1,j$', $1\leq j \leq k$, and (III) connecting a new root $r$ to the vertex (1). If $k>1$ we may omit the $\emptyset$ symbol, see Fig.~\ref{F_split_0}. Note that in general
\begin {equation*}
 \tau_1\ast\tau_2 \ast \cdots \ast \tau_k \neq  \tau_{\sigma(1)}\ast\tau_{\sigma(2)} \ast \cdots \ast \tau_{\sigma(k)} 
\end {equation*}
for a permutation $\sigma$ of $\{1,2,\ldots,k\}$.
\begin{figure} [!h]
\centerline{\scalebox{0.6}{\includegraphics{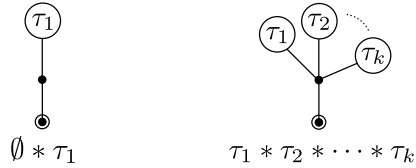}}}
\caption{The root joining operation.} \label{F_split_0}
\end{figure}

Let $\pi_n$ be a probability distribution on $\mathcal{T}_n$. We define a random tree $T_n$ by the canonical probability generating function
\begin {equation} \label{rt}
 T_n = \sum_{\tau \in \mathcal{T}_n} \pi_n(\tau) \tau.
\end {equation}
The above sum of trees and multiplication of trees by a scalar are formal and provide a convenient way of storing information on the probability measure $\pi_n$. 
\subsection{The ag--model}
Using the notation introduced above, we now define the ag--model which was described informally in the introduction. Let
\begin {equation}
 W(n) = n-a.
\end {equation}
 We introduce a growth operation $I$ in the following recursive way. Let $s$ be the single edge tree and define $I(s) = \emptyset \ast s$. For a tree $\tau = \emptyset \ast \tau_1 \ast \cdots \ast \tau_{k-1}$ define the random tree
\begin {eqnarray} \nonumber
 I(\tau) &=& \frac{1}{W(|\tau|)}  \Bigg( \sum_{i=1}^{k-1} W(|\tau_i|) ~\emptyset \ast \tau_1 \ast \cdots \ast I(\tau_i) \ast \cdots \ast \tau_{k-1} \\ \nonumber
&& +~ \alpha ~\emptyset \ast \tau + \frac{\overline{w}_k}{
k} \sum_{i=1}^k \tau_1 \ast \cdots \ast \tau_{i-1}\ast s \ast \tau_{i} \ast \cdots \ast \tau_{k-1} \Bigg) \\
\end {eqnarray}
where $|\tau|$ denotes the number of edges in $\tau$. The growth operation $I$ is equivalent to the growth rule which was described informally in Fig.~\ref{f:growth} (left) in the introduction. The ag--model is defined recursively as the random tree $P_{n}$ which satisfies $P_{1} = s$ and 
\begin{equation} \label{themodel}
P_{n} = I(P_{n-1}).
\end {equation}
We denote the probability measure on $\mathcal{T}_n$,  generated by this growth process by $\nu_{n}$.

\section {Markov branching} \label{s:mb}
A sequence of random trees $(T_n)_{n \geq 1}$  is said to satisfy a Markov branching property, or to be Markovian self--similar, if there exist functions $q_k(n_1,\ldots,n_{k-1})$, $k\geq 2$ such that for all $n \geq 2$
\begin {equation} \label{mbproperty}
 T_{n} = \sum_{k=2}^\infty \sum_{n_1+\cdots+n_{k-1} = n-1} q_k(n_1,\ldots,n_{k-1})~\emptyset \ast T_{n_1} \ast \cdots \ast T_{n_{k-1}}.
\end {equation}The functions $q_k(n_1,\ldots,n_{k-1})$ are referred to as the {\it first split distribution} of $(T_n)_{n\geq 1}$. We use the convention that $q_k(n_1,\ldots,n_{k-1}) = 0$ if any of the arguments $n_1,\ldots,n_{k-1}$ equals zero.
\begin {proposition} \label{p:MB1}
 The random trees $(P_{n})_{n\geq 1}$, defined by (\ref{themodel}), have the Markov branching property with a first split distribution which satisfies $q_2(1) = 1$,
\begin {equation}
 q_2(n) = \frac{1}{W(n)}\left(W(n-1) q_2(n-1) +\alpha\right),
\end {equation}
\begin {eqnarray} \nonumber
 q_k(n_1,\ldots,n_{k-1}) &=& \frac{1}{W(n)} \sum_{i=1}^{k-1}\Bigg(W(n_i-1) q_k(n_1,\ldots,n_i-1,\ldots,n_{k-1}) \\ \nonumber
&& + ~ \frac{\overline{w}_{k-1}}{k-1}  \delta_{n_i,1} q_{k-1}(n_1,\ldots,n_{i-1},n_{i+1},\ldots,n_{k-1})\Bigg) \\ \label{qrec2}
\end {eqnarray}
for $3 \leq k \leq D$ and
\begin {equation}
q_{D+1}(n_1,\ldots,n_D) = 0
\end {equation}
where $n_1 + \cdots + n_{k-1} = n$.
\end {proposition}


\begin{proof}
We use induction on $n$. $P_n$ clearly satisfies (\ref{mbproperty}) for $n=2$. Assume it satisfies (\ref{mbproperty}) for some $n$. Then
\begin {eqnarray*}
 P_{n+1} &=&  \sum_{k=2}^\infty \sum_{n_1+\cdots+n_{k-1} = n-1} q_{k}(n_1,\ldots,n_{k-1}) \frac{1}{W(n)} \Bigg({\alpha}~ \emptyset \ast(\emptyset\ast P_{n_1}\ast \cdots \ast P_{n_{k-1}})\\
&& +~ \frac{\overline{w}_k}{k} \sum_{i=1}^k  P_{n_{1}} \ast \cdots \ast P_{n_{i-1}}\ast s \ast P_{n_{i}} \ast \cdots \ast P_{n_{k-1}}  \\
&& +~\sum_{i=1}^{k-1} {W(n_i)} ~\emptyset \ast P_{n_1} \ast \cdots \ast I(P_{n_i}) \ast \cdots \ast P_{n_{k-1}}\Bigg)\\
&=& \frac{1}{W(n)}\left(\alpha + W(n-1)q_2(n-1)\right) \emptyset\ast P_n \\
&&~+ \sum_{k=3}^\infty \sum_{n_1+\cdots+n_{k-1} = n}\frac{1}{W(n)}\sum_{i=1}^{k-1} \Bigg(\frac{\overline{w}_{k-1}}{k-1} \delta_{n_i,1}q_{k-1}(n_1,\ldots,n_{i-1},n_{i+1},\ldots,n_{k-1}) \\
&& +~ W(n_i-1) q_{k}(n_1,\ldots,n_i-1,\ldots,n_{k-1})\Bigg) \emptyset\ast P_{n_1} \ast \cdots \ast P_{n_{k-1}}.\\
\end {eqnarray*}
This shows that (\ref{mbproperty}) also holds for $n+1$ and we conclude that it holds for all $n\geq 2$.
 \end {proof}

The recursions for the first split distribution in Proposition \ref{p:MB1} can be solved with straightforward methods. We state the result in the following proposition which can easily be proved by induction. The method for finding the solution is described in Appendix A.
\begin {proposition}
 The first split distribution of the sequence $(P_n)_{n\geq1}$ is given by
\begin {eqnarray} \nonumber
 q_k(n_1,\ldots,n_{k-1}) &=& \frac{ \Gamma\left(k-2+\frac{1-\alpha}{a}\right)}{\Gamma\left(\frac{1-\alpha}{a}\right)\Gamma(k)} \frac{\G{1-a}\G{n+1}}{a\Gamma\left(n+1-a\right)}\prod_{i=1}^{k-1} \frac{a\Gamma(n_i-a)}{\G{1-a}\G{n_i+1}}  \\ 
&& \times ~\left(1-a - \alpha + \alpha \sum_{i=1}^{k-1}\frac{n_i}{n+1-n_i}\right) \nonumber \\ \label{qfuncgengen}
\end {eqnarray}
where $n_1 + \cdots + n_{k-1} = n$.
\end {proposition}
We will repeatedly use the following standard, easily derived identities when we work with the above first split distribution \cite{abram}
\begin {equation} \label{sumofgamma}
 \sum_{n=1}^\infty \frac{c\Gamma(n-c)}{\G{1-c}\G{n+1}} z^{n} = 1-(1-z)^c
\end {equation}
and 
\begin {equation} \label{ordergamma}
 \frac{\Gamma(n-c)}{\G{n+1}} = n^{-c-1}\left(1+O\left(n^{-1}\right)\right).
\end {equation}

\section{Convergence of the finite volume measures} \label{s:cfvm}
In this section we show that the measures $\nu_n$ generated by the growth process converge weakly to a measure $\nu$ on the set of infinite trees. By weak convergence we mean that for all bounded functions $f$ which are continuous in the topology generated by the metric $d$
\begin {equation}
\int_{\mathcal{T}} f(\tau)d\nu_{n} \longrightarrow \int_{\mathcal{T}} f(\tau)d\nu, \quad\quad \text{as $n \longrightarrow \infty$.}
\end {equation}

 We will call an infinite non--backtracking path from the root, a {\it spine}. Let $\tau$ be a tree with exactly one spine and let $v$ be a vertex on the spine ($v\neq r$) with degree $k$. We call the $k-2$ finite subtrees of $\tau$ which are attached to the vertex $v$ {\it outgrowths} from the spine.

\begin {theorem} \label{th:conv}
Let $\alpha > 0$. The measures $\nu_{n}$, viewed as probability measures on $\mathcal{T}$, converge weakly, as $n \longrightarrow \infty$, to a probability measure $\nu$ which is concentrated on the set of trees that have exactly one spine. The degrees of the vertices on the spine are independently distributed by 
\begin {equation} \label{phiinfgen}
 \phi(k) =  \frac{\alpha \G{\frac{1+a}{a}}\G{k-2+\frac{1-\alpha}{a}}}{\G{\frac{1-\alpha}{a}}\G{k-1+\frac{1}{a}}}, \qquad  k \geq 2.
\end {equation}
The outgrowths from the spine are finite with probability one and outgrowths from different vertices are independently distributed. If a vertex $v$ on the spine has degree $k$ and $\tau_1,\ldots,\tau_{L}$ are the outgrowths from $v$ to the left of the spine (in that order) and $\tau_{L+1},\ldots,\tau_{k-2}$ are the outgrowths from $v$ to the right of the spine (in that order), then their joint distribution is 
\begin {equation} \label{muinfgen}
 \mu_k(\tau_1,\ldots,\tau_{k-2}) =   \frac{\Gamma\left(k-1+\frac{1}{a}\right)}{\Gamma\left(\frac{1+a}{a}\right)\Gamma(k-1)}\frac{1}{1+m}\prod_{i=1}^{k-2} \frac{a \Gamma(|\tau_i|-a)}{\G{1-a}\Gamma(|\tau_i|+1)}\nu_{|\tau_i|}(\tau_i), 
\end {equation}
 where $m = |\tau_1| + \cdots + |\tau_{k-2}| \geq k-2$ and $k\geq 3$. 
\end {theorem} \noindent
A proof to the above theorem is given on page \pageref{pr:conv} in Appendix B.


We point out that the distributions $\mu_k$ are independent of how many of the outgrowths are to the left or to the right of the spine. For an ordered sequence of $k-2$ outgrowths, there are $k-1$ different ways to arrange them around the spine.

Below, we comment on some special cases. When $D < \infty$, $\phi(k) = 0$ for $k > D$ as it should be. In the case $D=3$ and $\alpha = 1/2$ the trees are generic, i.e.~ $\nu_n$ is a critical Galton--Watson process conditioned to have $n$ edges. This follows from the fact that the first split distributions can, in this case, be written as 
\begin {equation}
q_2(n) = \frac{1}{2}\frac{Z_n}{Z_{n+1}} \qquad \text{and} \qquad q_3(n_1,n_2) = \frac{1}{4}\frac{Z_{n_1} Z_{n_2}}{Z_{n_1+n_2+1}}   
\end {equation}
with
\begin {equation}
 Z_n = \frac{\G{n+\frac{1}{2}}}{\sqrt{\pi}\G{n+2}}.
\end {equation}
$Z_n$ can be interpreted as a finite volume partition function corresponding to branching weights $w_1 = w_3 = 1/4$ and $w_2 = 1/2$, see e.g.~\cite{sdgt}. Furthermore, this is the only special case in which we obtain generic trees. This can be seen from the fact that when $D>3$, outgrowths from the same vertex on the spine are dependent.

When $D=\infty$ and $\gamma = 1$, 
\begin {equation} \label{eq:degexp}
 \phi(k) = \alpha (1-\alpha)^{k-2} 
\end {equation}
and it falls off exponentially in $k$. When $D=\infty$ and $\gamma < 1$ we find that for large $k$
\begin {equation} \label{eq:degpower}
\phi(k) = \frac{\alpha \G{\frac{2-\gamma}{1-\gamma}}}{\G{\frac{1-\alpha}{1-\gamma}}} k^{-1-\frac{\alpha}{1-\gamma}}\left(1+O\left(k^{-1}\right)\right) 
\end {equation}
i.e.~it falls off with a power law in $k$. From the last formula, we see that when $\alpha \leq 1-\gamma$, the expected value of the degree of a vertex on the spine is infinite. A simple and interesting special case arises when $\gamma = 0$ in which case the outgrowths from the spine are single leaves. Such graphs have been referred to as {\it caterpillars} in the literature. The degrees of the vertices on the spine are distributed independently by
\begin {equation} \label{phiinfcat}
 \phi(k) = \frac{\alpha \Gamma\left(k-1-\alpha\right)}{\Gamma\left(k\right)\Gamma\left(1-\alpha\right)} 
\end {equation}
and they have an infinite expected value for all values of $\alpha$.   These caterpillars are a special case of 'caterpillars at a phase transition' in the equilibrium statistical mechanical model studied in \cite{jonsson:2009}. We will consider this special case in more detail in the next section.
\section{The Hausdorff dimension} \label{s:hd}
The Hausdorff dimension is a notion of dimension of graphs and is defined in terms of how the volume of the graph ball $B_R$ scales with its radius $R$. The Hausdorff dimension of a graph $G$ is defined as 
\begin {equation} \label{dhdef}
d_H = \lim_{R\rightarrow\infty} \frac{\log(|B_R(G)|)}{\log(R)}
\end {equation}
provided that the limit exists. This definition is only interesting on an infinite graph. On the hyper--cubic lattice $\mathbb{Z}^d$ it holds that $d_H = d$ but in general $d_H$ is not an integer. This dimension has been studied by physicists, especially in the quantum gravity literature, see e.g.~\cite{book} and should not be confused with the usual notion of Hausdorff dimension in a metric space, although there are some similarities.

The Hausdorff dimension can be defined in different ways for random graphs. If the graphs are distributed by $\nu$ then they might first of all have, $\nu$--almost surely, a Hausdorff dimension $d_H$ as defined above. Secondly, we define the annealed Hausdorff dimension as 
\begin {equation} \label{dhdefexp}
\bar{d}_H =  \lim_{R\rightarrow\infty} \frac{\log(\langle|B_R(G)|\rangle_\nu)}{\log(R)}
\end {equation}
where $\langle \cdot \rangle_\nu$ denotes expected value with respect to $\nu$.

There is another notion of dimensionality which applies when one considers a sequence of finite volume measures $(\nu_n)_{n\geq 0}$ on a set of graphs. It is usually defined in terms of how the average value of some typical distance in the graph (the maximum distance between vertices, the mean distance of vertices from the root, etc.) scales in relation to the volume of the graph $n$ as it grows. This dimension has also been referred to as the Hausdorff dimension in the physics literature but to avoid confusion we will refer to it here as the {\it fractal dimension} and denote it by $d_f$.  To give a more precise definition, we adopt the one from \cite{vs} which is as follows: Define the radius of a finite tree $T$ by
\begin {equation}
 R_T = \frac{1}{2|T|} \sum_{v \in V(T)} d_T(r,v) \sigma_T(v)
\end {equation}
where $V(T)$ is the vertex set of $T$, $r$ is the root, $d_T$ is the graph metric and $\sigma_T(v)$ denotes the degree of $v$. The fractal dimension is defined as
\begin {equation} \label{eq:fractal}
 d_f = \lim_{n\rightarrow\infty} \frac{\log(n)}{\log(\langle R_T\rangle_{\nu_n})}.
\end {equation}
 If $\nu_n$ converge to a measure $\nu$ concentrated on infinite graphs, $d_f$ has been observed to be equal to $d_H$ (or $\bar{d}_H$) in many situations, a simple example is the uniform tree and modifications of it, see e.g.~\cite{sdgt}. It is however straightforward to find a counterexample where $d_H \neq d_f$ and it is not entirely clear which conditions guarantee equality. We will comment on this relation in the ag-model below.

We will now calculate the annealed Hausdorff dimension of the trees distributed by $\nu$ from Theorem \ref{th:conv}. 
\begin {theorem}\label{th:annhaus}
Let  $\max\{0,a\} < \alpha \leq 1$. The random trees, distributed by $\nu$ described in Theorem \ref{th:conv}, have an annealed Hausdorff dimension
\begin {equation}
 \bar{d}_H = \frac{1}{\alpha}.
\end {equation}
\end {theorem}
To prove Theorem \ref{th:annhaus}, we need to analyse the large $R$ behaviour of $\langle |B_R| \rangle_{\nu}$.  In order to simplify the notation we let $(\emptyset)$ be the empty tree and define $\mu_2((\emptyset)) = 1$. We then extend the probability distributions $\mu_k$, $k \geq 2$, to probability distributions on
$ \bigcup_{k=2}^\infty \mathcal{T}^{k-2}$ and define
\begin {equation} \label{mudef}
 \mu = \sum_{k=2}^\infty \phi(k) \mu_k.
\end {equation}
Since the outgrowths from different vertices on the spine are i.i.d.~it is clearly sufficient to show that 
\begin {equation}
 \Big\langle \sum_{i} |B_R(\tau_i)| \Big\rangle_{\mu} = R^{1/\alpha - 1} (1+o(1))
\end {equation}
as $R\longrightarrow \infty$. This follows from the Lemma below which is proved on page \pageref{pr:annhaus} in Appendix B.
\begin {lemma} \label{l:annhaus}
For  $\max\{0,a\} < \alpha \leq 1$,
\begin {eqnarray*} 
  \Big\langle \sum_{i} |B_R(\tau_i)| \Big\rangle_{\mu} &=& \frac{(R + \frac{\alpha - a}{\alpha})\G{\frac{\alpha-a}{\alpha}} \G{R+\frac{1-a}{\alpha}}}{\G{\frac{1-a}{\alpha}}\G{R+\frac{2\alpha -a}{\alpha}}}-1.
\end {eqnarray*}
\end {lemma}

Note that $\bar{d}_H = \infty$ when $D = \infty$ and $\alpha \leq 1-\gamma$ ($\alpha \leq \max\{0,a\}$) since then the expected value of degrees of vertices on the spine is infinite. However, the $\nu$--almost sure Hausdorff dimension might still be finite. We confirm this in the case $D=\infty$ and $\gamma = 0$, when the trees are caterpillars.
\begin {theorem} \label{th:ashaus}
 Let $0 < \alpha \leq 1$ and $\gamma = 0$. Then
\begin {equation}
 d_H = \frac{1}{\alpha}
\end {equation}
$\nu$--almost surely.
\end {theorem} \noindent
The proof is given on page \pageref{pr:ashaus} in Appendix B.
\subsection{Comparison to the fractal dimension}
In the original paper on the vertex splitting model \cite{vs} it was shown that the expected value of the radius of a tree can be written as
\begin {equation} \label{eq:radius}
 \langle R_T \rangle_{\nu_n} = \frac{n+1}{2n}\sum_{n_2=0}^n (2n_2+1)\sum_{k} \tilde{q}_k(n-n_2,n_2)
\end {equation}
where $\tilde{q}_k(n_1,n_2)$ is the probability that a uniformly chosen vertex $v$ has degree $k$ and that the volume of the subtree attached to $v$ containing the root is $n_1$ and that the other subtrees attached to $v$ have a total volume $n_2$. Furthermore, in the case of linear splitting weights $w_i = Ai+B$,  $\tilde{q}_k(n_1,n_2)$ was shown to be a solution of a system of linear recursion equations determined by the growth rules of the vertex splitting model, see \cite[Section 3]{vs}. These recursion equations could not be solved explicitly but it was assumed that the following scaling holds
\begin {equation} \label{eq:2pscaling}
 \tilde{q}_k(n_1,n_2) \sim (n_1+n_2)^{-2+\lambda}\omega_k(n_1/(n_1+n_2))
\end {equation}
for some $\lambda$ and  ``scaling functions`` $\omega_k$.  The linear recursions were thus reduced to an eigenvalue equation for $\lambda$
\begin {equation} \label{eq:eig}
 \mathbf{C} \mathbf{x} = \lambda \mathbf{x}
\end {equation}
where $\lambda$ is the Perron--Frobenius eigenvalue of the $\binom{D}{2} \times \binom{D}{2}$ matrix $\mathbf{C}$ indexed by a pair of two indices $ki$, $2 \leq k \leq D$, $2 \leq i < k$, and given by the matrix elements
\begin {equation} \label{eigeq}
 C_{ki,jn} = \frac{1}{w_2}(w_{k,j+2-k}((j-i)\delta_{i,n} + i \delta_{n,j-k+i}) - w_k \delta_{k,j}\delta_{i,n}).
\end {equation}
Comparing (\ref{eq:radius}) and (\ref{eq:2pscaling}) to (\ref{eq:fractal}) allows one to find the fractal dimension
\begin {equation} \label{eq:fracresult}
 d_f = \lambda^{-1}.
\end {equation}
The scaling assumption (\ref{eq:2pscaling}) was not proven but the results (\ref{eq:eig}-\ref{eq:fracresult}) were supported by simulations in the case $D=3$.

It is interesting to compare $d_f$, corresponding to the weights (\ref{partw}) of the ag--model, to the values of $\bar{d}_H$ obtained in Theorem \ref{th:annhaus}. It is straightforward to solve (\ref{eq:eig}) for small values of $D$ and find that $d_f = \bar{d}_H = 1/\alpha$. Furthermore, we have calculated $d_f$ in the case $D=\infty$ and $\gamma = 1$ by solving (\ref{eq:eig}) numerically. We used a cutoff $D=30$ on the system which is expected to closely approximate the case $D=\infty$, since the vertex degree distribution is believed to fall of exponentially in this case, cf.~(\ref{ddist}). The results are shown in Fig.~\ref{f:hausinf}. The agreement we find, supports the validity of the scaling assumption (\ref{eq:2pscaling}) to a very high maximum degree $D$.
\begin{figure} [!h]
\centerline{\scalebox{0.3}{\includegraphics{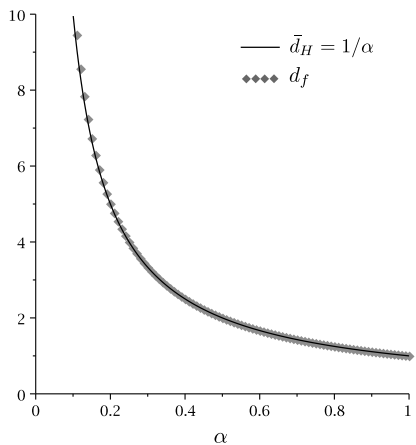}}}
\caption{Comparison of $d_f$ (gray squares) and $\bar{d}_H = 1/\alpha$ (solid line) in the case $D=\infty$ and $\gamma = 1$. Using the weights (\ref{partw}), we calculated $d_f$ numerically from (\ref{eq:eig}-\ref{eq:fracresult}) for $\alpha = i/100,~1\leq i \leq 100$, using a cutoff $D=30$ on the system.} \label{f:hausinf}
\end{figure}

\subsection{The spectral dimension}

We conclude this section by mentioning another notion of dimension of graphs called the {\it spectral dimension}. It is defined in terms of how the return probability of a random walker on the graph decays with time $t$. More precisely, for a tree $\tau$ let $p_\tau(t)$ be the probability that a simple random walk which leaves the root at time $t=0$ is back at the root at time $t$. The spectral dimension of $\tau$ is defined as
\begin {equation}
d_s = -\lim_{t\rightarrow\infty}\frac{2\log(p_\tau(t))}{\log(t)}
\end {equation}
 provided the limit exist. The spectral dimension can take any value greater than one and does not necessarily agree with the Hausdorff dimension. We refer to \cite{barlow:2005,sdgt,combs,brushes} for discussion of the spectral dimension of several types of random graphs.

It would be interesting to calculate the spectral dimension of the trees distributed by $\nu$. For now we only have results in the case when the trees are caterpillars.

\begin {theorem} \label{th:spec}
Let $0 < \alpha \leq 1$ and $\gamma = 0$. If $d_s$ exists then
\begin {equation}
 d_s = \dfrac{2}{1+\alpha}
\end {equation}
$\nu$-- almost surely
 \end {theorem} \noindent
The theorem is proved on page \pageref{pr:spec} in Appendix B. 

\section{Vertex degree distribution and correlations}
In this section we use results from the vertex splitting model \cite{vs} to calculate the vertex degree distribution and correlation between the degrees of neighbouring vertices in the ag--model. Not all results in this section are rigorous and we will comment on this point below.

Let $X_{i,n}$ be the number of vertices of degree $i$ in a random tree with $n$ edges and define the vertex degree densities
\begin {equation}
 \rho_i = \lim_{n\rightarrow\infty} \frac{\mathbb{E}(X_{i,n})}{n}.
\end {equation}
It was shown in \cite[Section 2]{vs} that the densities in the vertex splitting model satisfy the linear equation
\begin{equation} \label{mf}
 \rho_i = -\frac{w_i}{w_2}\rho_i + \sum_{j\geq k-1} \frac{j w_{i,j+2-i}}{w_2}\rho_j
\end{equation}
assuming that the splitting weights are linear $w_i = Ai+B$ and under certain technical conditions on the partitioning weights. The splitting weights are linear in the ag--model, cf.~(\ref{splittingweights}), and one can check that for small $D$ the technical conditions needed on the partitioning weights are fulfilled. However, it is not certain whether (\ref{mf}) holds for general $D$ and one would need further analysis to verify that.  It is straightforward to solve (\ref{mf}) for the weights (\ref{partw}) and thus we find that the vertex degree densities in the ag--model are given by
\begin {equation}
 \rho_1 = \frac{1-\alpha}{2-a-\alpha}
\end {equation}
and
\begin {equation} \label{ddist}
 \rho_k = \frac{(1-a)\G{\frac{a+2-\alpha}{a}}\G{k-2+\frac{1-\alpha}{a}}}{(2-a-\alpha)(2-\alpha)\G{\frac{1-\alpha}{a}}\G{k-1+\frac{2-\alpha}{a}}}, \qquad {k\geq 2}
\end {equation}
provided that (\ref{mf}) holds. By sending $\alpha$ to zero we find that these results agree with results previously obtained in the preferential attachment model \cite{albert0,rapoport}. Also note, that in the case $D=\infty$, $\rho_k$ has in general a power law behaviour except when $\gamma = 1$ ($a=0$) in which case it falls of exponentially with rate $(1-\alpha)/(2-\alpha)$. This resembles properties of the degree distribution of the vertices on the spine, cf.~(\ref{eq:degexp}) and (\ref{eq:degpower}).

Let $X_{ij,n}$ be the number of edges with endpoints of degree $i$ and $j$ in a random tree with $n$ edges, using the convention that the vertex of degree $i$ is the one closer to the root. Define the density
\begin  {equation} \label{limitcorr}
 \rho_{ij} = \lim_{n\rightarrow\infty}\frac{\mathbb{E}(X_{i,j,n})}{n}.
\end  {equation}
It was shown in \cite{vs} that these densities in the vertex splitting model satisfy
\begin {eqnarray} \nonumber
 \rho_{jk} &=& -\frac{w_j+w_k}{w_2}\rho_{jk} + (j-1)\frac{w_{j,k}}{w_2}\rho_{j+k-2} + (j-1)\sum_{i\geq j-1} \frac{w_{j,i+2-j}}{w_2} \rho_{ik}\\
 && +~ (k-1)\sum_{i\geq k-1} \frac{w_{k,i+2-k}}{w_2}\rho_{ji} \label{rhojk}
\end {eqnarray}
assuming that the limit (\ref{limitcorr}) exists. The densities give us information about correlation between vertex degrees of neighbouring vertices. It can be measured with a correlation coefficient
\begin {equation}
 r = \frac{\sum_{j,k}(j-1)(k-1)(\bar{\rho}_{jk}-\bar{\rho}_j\bar{\rho}_k)}{\sum_{j}(j-1)^2\bar{\rho}_j - \left(\sum_{j}(j-1)\bar{\rho}_j\right)^2}
\end {equation}
where
\begin {equation}
 \bar{\rho}_k = \frac{k\rho_k}{\sum_i i\rho_k} \qquad \text{and} \qquad \bar{\rho}_{ij} = \frac{\rho_{ij}+\rho_{ji}}{2}.
\end {equation}
The coefficient $r$ takes values between $-1$ and $1$. If $r < 0$ the graph is said to show disassortative mixing and vertices with high degree prefer to be neigbours of vertices with low degree. If $r > 0$ the graphs are said to show assortative mixing and vertices with high degree prefer to be neighbours of vertices with high degree, see e.g.~\cite{newman}.

We will conlude this section by calculating $r$ for two choices of parameters in the ag--model, namely $D=3$ and $D=\infty$, $\gamma = 1$. We consider (\ref{rhojk}) with the weights given in (\ref{partw}). In the case $D=3$, (\ref{rhojk}) can be explicitly solved and we find that
\begin {equation} \label{r1}
r = -\frac{\alpha(4-3\alpha)}{(5-3\alpha)(3-\alpha)},	\qquad (D=3).
\end {equation}
This can of course be repeated for small values of $D$. However, when $D$ is large or infinite it is more difficult to solve (\ref{rhojk}) explicitly. Instead, we study the generating function
\begin {equation}
 S(x,y) = \sum_{j,k\geq 2} \bar{\rho}_{jk} x^{j-1} y^{k-1}
\end {equation}
and use the fact that
\begin {equation}
 \partial_x\partial_y S(1,1) = \sum_{j,k}(j-1)(k-1)\bar{\rho}_{jk}
\end {equation}
to calculate $r$. Equation (\ref{rhojk}) becomes a linear, first order partial differential equation in terms of the generating function $S(x,y)$. It can in principle be solved for a general set of parameters, however we only comment on the case $D=\infty$, $\gamma = 1$. In that case, the coefficients of the derivative terms in the PDE are zero and we get an ordinary equation for $S(x,y)$. The solution is
\begin {equation}
S(x,y) = \frac{\sigma (x^2 f(y)+y^2 g(x))}{1- \sigma x - \sigma y} + x f(y) + y g(x) + \rho_{22} x y
\end {equation}
where
\begin {equation}
 f(y) = \sum_{k=3}^\infty \rho_{2k} y^{k-1} \quad \text{and} \quad g(x) = \sum_{j=3}^\infty \rho_{j2} x^{j-1},
\end {equation}
\begin {eqnarray} \nonumber
\rho_{2k} &=& \frac{2(1-3\alpha + \alpha^2)}{(1-\alpha)(2-\alpha)^2} \sigma^k+\frac{\alpha}{(1-\alpha)^2} \eta^k,\qquad  k\geq 2, \\ \nonumber
\rho_{j2} &=& -\frac{2\alpha^4-15\alpha^3+38\alpha^2-37\alpha+8+(2\alpha^2-6\alpha +4)j}{(1-\alpha)^2(2-\alpha)^2}  \sigma^j \\ \nonumber
&& +~  \frac{2\alpha^2-6\alpha +2 + \alpha j}{(1-\alpha)^2} \eta^j, \qquad  j\geq 3,
\end {eqnarray}
with
\begin {equation}
 \eta = \frac{1-\alpha}{2-\alpha} \quad \text{and} \quad \sigma = \frac{1-\alpha}{3-\alpha}.
\end {equation}
From these expressions we find that
\begin {equation} \label{r2}
 r = -\frac{\alpha}{2(1+\alpha)} \qquad (D=\infty, \gamma = 1).
\end {equation}
We plot the solutions (\ref{r1}) and (\ref{r2}) together in Fig.~\ref{f:plot}.
\begin{figure} [!h]
\centerline{\scalebox{0.20}{\includegraphics{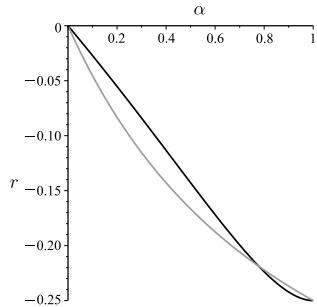}}}
\caption{Comparison of Equations (\ref{r1}) ($D=3$, black) and (\ref{r2}) ($D=\infty$, $\gamma=1$, gray).} \label{f:plot}
\end{figure}
The two curves are similar in both cases. If $\alpha = 0$ then $r=0$ which agrees with results which have previously been obtained for the preferential attachment model \cite{newman}. In this case the vertices which are close to the root are 'old' in the sense that once they reach the maximum degree (which they eventually do with probability one) they do not change again. Thus, a lot of vertices of high degree become neigbours. When $\alpha$ is increased above zero, a repulsions is introduced between these vertices, the value of $r$ decreases and the trees show disassortative mixing. When $\alpha$ goes to 1, the trees approach the same non--random graph, a spine with no outgrowths. As a consequence, the value of $r$ approaches the same value in both cases.  
\section{Conclusions}
We introduced the ag--model, a special case of the vertex splitting model which has the Markov branching property. For particular choices of parameters it reduces to models of generic trees \cite{sdgt}, preferential attachment \cite{albert0} and non--generic caterpillars \cite{jonsson:2009}. It was proved that the finite volume measures generated by the growth rules converge to a measure which is concentrated on the set of trees with exactly one spine and the limiting measure was described explicitly. The same has been done before in Ford's $\alpha$--model \cite{siggi} and a special case of the $\alpha \gamma$--model \cite{phdsos}. Extension of these convergence results to the vertex splitting model is a work in progress. 

There is another notion of convergence of random trees, referred to as the {\it scaling limit}, see e.g.~\cite{aldousoriginal,continuum}. This means, roughly, that a random tree $T_n$ viewed as a metric space with the graph metric $d_{\text{gr}}$ suitably scaled, converges weakly, in the Gromov--Hausdorff topology, to a continuum random tree.  In a recent paper on Markov branching trees \cite{markovss}, Haas and Miermont proved that under certain natural conditions on the first split distributions the scaling limit of the trees is a self--similar fragmentation tree, in the Gromov--Hausdorff-Prokhorov topology. We expect that this theory applies to the model studied in this paper and it would be interesting to confirm that. Moreover, it is an interesting and challenging problem to generalize the results on the scaling limit to the vertex splitting model when Markov branching is absent.

 The annealed Hausdorff dimension, with respect to the infinite volume measure of the ag--model, was calculated for a certain range of the parameters.  The results partly support scaling assumptions which were made when calculating the fractal dimension in the vertex splitting model \cite{vs}. In the special case of growing caterpillar graphs we calculated, almost surely, the Hausdorff and spectral dimension. It turns out that the dimensions are related by the formula
\begin {equation}
 d_s = \frac{2 d_H}{1+d_H}.
\end {equation}
This equation holds in general for tree models which satisfy a certain uniformity condition and under the assumption that vertex degrees are uniformly bounded from the above \cite{barlow:2005}. We expect this relation to hold in the ag--model and it would be desirable to check whether it holds in the vertex splitting model.

It is possible to study other interesting observables in the ag--model such as the vertex degree distribution and correlations between degrees of neighbouring vertices in large trees. It would be interesting to give a rigorous proof of (\ref{ddist}) and even to get stronger convergence results for the random variables $X_{i,n}$. This can presumably be done, at least for some range of the parameters, using results on generalized Pólya urns \cite{svante}. Furthermore, it would be interesting to confirm the validity of (\ref{mf}) for as general set of parameters as possible. Similar results about the convergence of $X_{ij,n}$ are desirable.

 A natural question is whether the ag--model is the only special case of the vertex splitting model which has the Markov branching property. As was noted in the introduction, the $\alpha \gamma$--model has the Markov branching property but it is not strictly a special case of the vertex splitting model, rather a limiting case. Since the vertex splitting model has local and isotropic growth rules, one might also ask whether there exists some other notion of self--similarity which could be used to handle the general case. An understanding of this could be a key element towards a solution of the most general case.

\begin {ack}
 I am deeply indebted to Thordur Jonsson and François David for helpful discussions and comments.
\end {ack}

 \appendix
\section*{Appendices}
\section {Solution of the first split distributions}
In this section we describe a 'network flow method' for solving the recursion equations for the first split distribution given in Proposition \ref{p:MB1}. We encountered this method in \cite{Ford} where it was used to solve recursion equations for the first split distribution of Ford's $\alpha$--model. 

First of all, it is straightforward to derive (\ref{qfuncgengen}) in the case $k=2$. Consider next the case $k=3$ in which case the nearest neighbour of the root has two disjoint subtrees of descendants which we refer to as the left and right subtree. We represent a state when the tree has $n_1$ edges in the left subtree and $n_2$ edges in the right subtree by a node $(n_1,n_2)$, $n_1,n_2 \geq 1$ in the network in Fig.~\ref{f:network}. 
\begin{figure} [!h]
\centerline{\scalebox{0.75}{\includegraphics{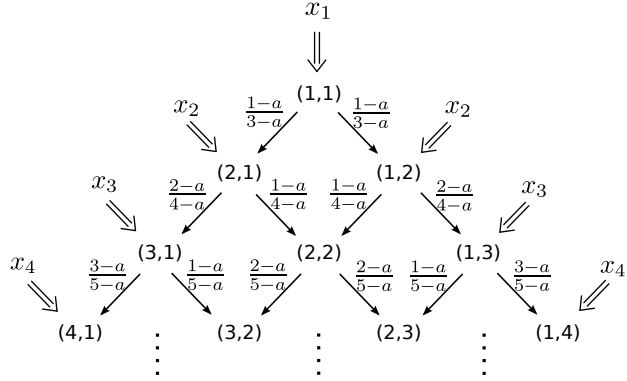}}}
\caption{A network flow diagram with sources $x_i$, $i\geq1$.} \label{f:network}
\end{figure}
We assign {\it conductance }
\begin {equation}
C((n_1-1,n_2) \rightarrow (n_1,n_2)) = \frac{n_1-1-a}{n_1+n_2-a}
\end {equation}
between states $(n_1-1,n_2)$ and $(n_1,n_2)$, which is the probability of going from state $(n_1-1,n_2)$ to state $(n_1,n_2)$ (symmetric in $n_1$ and $n_2$) given by the recursion (\ref{qrec2}).  Let $\omega = (\omega_0, \omega_1,\ldots,\omega_{|\omega|})$ be a simple path (non--backtracking path) in the network with endpoints $\omega_0$ and $\omega_{|\omega|}$ where $|\omega|$ denotes the length of the path. Note that the product of conductance along a path $\omega$ from a state $(i,1)$ (or $(1,i)$) to a state $(n_1,n_2)$ is given by
\begin {equation}
 \prod_{j=1}^{|\omega|} C(\omega_{j-1}\rightarrow\omega_{j}) = \frac{ (i-a)(i+1-a)\G{n_1-a}\G{n_2-a}}{\G{n_1+n_2+1-a}\G{1-a}}
\end {equation}
and is independent of the path chosen. 
We define the {\it value} of a state $(n_1,n_2)$ as $q_3(n_1,n_2)$ and we define $x_{n_1}$ as a source which flows into state $(n_1,1)$ (or $(1,n_1)$). We can write
\begin {eqnarray} \nonumber
 q_3(n_1,n_2) &=& C((n_1-1,n_2) \rightarrow (n_1,n_2)) q_3(n_1-1,n_2)(1-\delta_{n_1,1}) \\ \nonumber 
&+& C((n_1,n_2-1) \rightarrow (n_1,n_2)) q_3(n_1,n_2-1) (1-\delta_{n_2,1}) \\ \nonumber
&+&x_{n_1} \delta_{n_2,1} + x_{n_2} \delta_{n_1,1} \\ \label{flowrelation}
\end {eqnarray}
i.e.~$q_3(n_1,n_2)$ is given by the flow from the neighbouring sources and the neighbouring states, weighted by the conductance between the states. By comparing (\ref{flowrelation}) to (\ref{qrec2}) we find that 
\begin {equation}
x_1 = q_3(1,1) = 1-q_2(2) = \frac{1-\alpha}{2-a} 
\end {equation}
and
\begin {equation}
x_{n_1} = \frac{\overline{w}_2 q_2(n_1)}{2 W(n_1+1)}=  \frac{(1-\alpha)(\alpha n_1 - \alpha + 1 -a)}{2(n_1+1-a)(n_1-a)}
\end {equation}
for $n_1 \geq 2$.

From (\ref{flowrelation}) we conclude that $q_3(n_1,n_2)$ is given by the sum over all paths from all sources which lead to the state $(n_1,n_2)$, where each path is weighted by the product of the conductance along the path, i.e.~
\begin {eqnarray} \nonumber
 q_3(n_1,n_2) &=&  \sum_{i=1}^{n_1} x_i \sum_{\substack{\text{paths $\omega$ from $(i,1)$} \\ \text{to $(n_1,n_2)$}}} \prod_{j=1}^{|\omega|} C(\omega_{j-1}\rightarrow\omega_{j}) \\  &+& \sum_{i=2}^{n_2} x_i \sum_{\substack{\text{paths $\omega$ from $(1,i)$} \\ \text{to $(n_1,n_2)$}}} \prod_{j=1}^{|\omega|} C(\omega_{j-1}\rightarrow\omega_{j}) \label{qsum}
\end {eqnarray}

Since the conductance along a path between given states is independent of the path chosen, we can take the product outside the inner sum  and we are simply left with a counting problem. The number of paths between $(i,1)$ and $(n_1,n_2)$ is $\binom{n_1+n_2-i-1}{n_1-i}$ and the number of paths between $(1,i)$ and $(n_1,n_2)$ is $\binom{n_1+n_2-i-1}{n_2-i}$. We can now easily perform the sums over $i$ and we recover (\ref{qfuncgengen}) for $k=3$. 

This argument can be generalized to higher values of $k$ and it yields the formula (\ref{qfuncgengen}).

\section {Proof of main theorems}

In this section we collect proofs of theorems and lemmas stated in the main part of the paper. We need the following two lemmas in the proof of Theorem \ref{th:conv}. 

\begin {lemma} \label{l:explicit}
For $n\geq k$,
\begin {eqnarray} \nonumber
  && \sum_{n_1+\cdots+n_{k-1} = n-1} q_k(n_1,\ldots,n_{k-1})=  \frac{\alpha \G{\frac{1+a}{a}}\G{k-2+\frac{1-\alpha}{a}}}{\G{\frac{1-\alpha}{a}}\G{k-1+\frac{1}{a}}} \\  \nonumber
&& +~ \frac{\G{1-a}\G{k-2+\frac{1-\alpha}{a}}}{\G{\frac{1-\alpha}{a}}\G{n-a}} \sum_{i=1}^{k-1} \left(1-a-\alpha+\frac{\alpha a i}{a i -a +1}\right) \frac{(-1)^{i+1} \G{n-1-a i}}{\G{i}\G{k-i}\G{1-a i}}. \\ 
\label{longq}
\end {eqnarray}
\end {lemma}
\begin {proof}
The proof follows from a generating function argument. Define
\begin {equation}
 A_{n,k} = \sum_{\substack{n_1+\cdots+n_{k-1} = n-1\\n_i \geq 1}} \prod_{i=1}^{k-1} \frac{a \Gamma(n_i-a)}{\G{1-a} \G{n_i+1}} \left(1-a-\alpha + \alpha \sum_{i=1}^{k-1} \frac{n_i}{1+\sum_{j\neq i} n_j}\right).
\end {equation}
Then
\begin {equation} \label{qanda}
  \sum_{n_1+\cdots+n_{k-1} = n-1} q_k(n_1,\ldots,n_{k-1})=  \frac{\G{k-2+\frac{1-\alpha}{a}}\G{1-a}\G{n}}{\G{\frac{1-\alpha}{a}}\G{k}a\G{n-a}}  A_{n,k}.
\end {equation}
By (\ref{sumofgamma}),
\begin {eqnarray} \nonumber
 \sum_{n=k}^\infty A_{n,k}\zeta^n &=& \zeta \Big((1-a-\alpha) \left(1-(1-\zeta)^{a}\right)^{k-1}\\ \nonumber 
&& +~ a \alpha (k-1)(1-\zeta)^{a-1} \int_0^\zeta \left(1-(1-\zeta')^{a}\right)^{k-2} d\zeta'\Big) \\ \nonumber
&& = \zeta\Big(\frac{\alpha a \G{k}\G{\frac{1+a}{a}}}{\G{k-1+\frac{1}{a}}} (1-\zeta)^{a-1} \\ \nonumber
&& + ~ \sum_{i=0}^{k-1}  \left(1-a-\alpha + \frac{\alpha a i}{a i - a + 1}\right)\binom{k-1}{i}(-1)^{i} (1-\zeta)^{a i}\Big).\\ \label{genank}
\end {eqnarray}
From the last expression we  can determine the coefficients $A_{n,k}$ and inserting them into (\ref{qanda}) completes the proof.
\end {proof}

\begin {lemma} \label{l:dom}
For $\alpha > 0$,
\begin {eqnarray} \nonumber
&&\lim_{n\rightarrow \infty} \sum_{k=2}^\infty  \sum_{n_1+\cdots+n_{k-1} = n-1} q_k(n_1,\ldots,n_{k-1}) \\ 
&& = \sum_{k=2}^\infty \lim_{n\rightarrow \infty} \sum_{n_1+\cdots+n_{k-1} = n-1} q_k(n_1,\ldots,n_{k-1}) = 1. \label{sumlimitqgengen}
\end {eqnarray}

\end {lemma} 
\begin {proof}
We obtain the following limit from (\ref{longq})
\begin {equation} \label{limitqgengen}
\lim_{n\rightarrow \infty} \sum_{n_1+\cdots+n_{k-1} = n-1} q_k(n_1,\ldots,n_{k-1}) =   \frac{\alpha \G{\frac{1+a}{a}}\G{k-2+\frac{1-\alpha}{a}}}{\G{\frac{1-\alpha}{a}}\G{k-1+\frac{1}{a}}}
\end {equation}
 since the second term in (\ref{longq}) converges to zero as $n\rightarrow\infty$. The second equality in (\ref{sumlimitqgengen}) then follows in a straightforward way (e.g.~by representing the sum of the right hand side of (\ref{limitqgengen}) by a hypergeometric function, see for instance \cite{abram}). The first equality is then trivial since by definition
\begin {equation}
 \sum_{k=2}^\infty  \sum_{n_1+\cdots+n_{k-1} = n-1} q_k(n_1,\ldots,n_{k-1}) = 1.
\end {equation}
  
\end {proof}

\begin {proof}[Proof of Theorem \ref{th:conv}] \label{pr:conv}
Due to properties of the metric space $(\mathcal{T},d)$ (see \cite{bergfinnur},\cite{ngtrees}), it is sufficient to show that the sequence
\begin {equation} \label{convseq}
\big(\nu_n(\{\tau \in \mathcal{T} : \tau_0 \text{ is a left subtree of $\tau$}\})\big)_{n\in \mathbb{N}}
\end {equation}
converges for any finite tree $\tau_0$. We will write $\nu_n(\tau_0)$ as a shorthand for (\ref{convseq}). We proceed by induction on $|\tau_0|$. Since every tree has the single rooted edge $s$ as a left subtree, (\ref{convseq}) clearly converges when $|\tau_0|=1$. Assume that it converges for all finite trees $\tau$ for which $|\tau| \leq N$ and denote the limits by $\nu(\tau)$. Take a tree $\tau_0$ such that $|\tau_0| = N+1$.  Denote the degree of the nearest neighbour of the root in $\tau_0$ by $k$ and decompose $\tau_0$ as $\tau_0 = \emptyset \ast \tau_1 \ast \cdots \ast \tau_{k-1}$. To simplify notation, we define $\tau_i = s$ for $i\geq k$, such that $\nu_n(\tau_i) = 1$ when $i \geq k$. From the Markov branching property we find that
\begin {eqnarray} \nonumber
\nu_{n}(\tau_0) &=& \sum_{\ell = k}^{\infty}\sum_{\substack{n_1 + \cdots + n_{\ell-1} = n-1}}q_\ell(n_1,\ldots,n_{\ell-1})  \prod_{j=1}^{\ell-1} \nu_{n_j}(\tau_j) \\\label{limgengenstep1}
\end {eqnarray}
see Fig.~\ref{f:measure}. We need to take the $n\rightarrow \infty$ limit of the above equation. First note that
\begin {eqnarray*}
 &&\sum_{\substack{n_1 + \cdots + n_{\ell-1} \\ = n-1}}q_\ell(n_1,\ldots,n_{\ell-1}) \prod_{j=1}^{\ell-1} \nu_{n_j}(\tau_j) \leq \sum_{\substack{n_1+\cdots+n_{\ell-1} \\ = n-1}}q_\ell(n_1,\ldots,n_{\ell-1})
\end {eqnarray*}
and thus, by Lemma \ref{l:dom} and dominant convergence we can swap the $n\rightarrow\infty$ limit and the sum over $\ell$ in (\ref{limgengenstep1}). 
\begin{figure} [!h]
\centerline{\scalebox{0.6}{\includegraphics{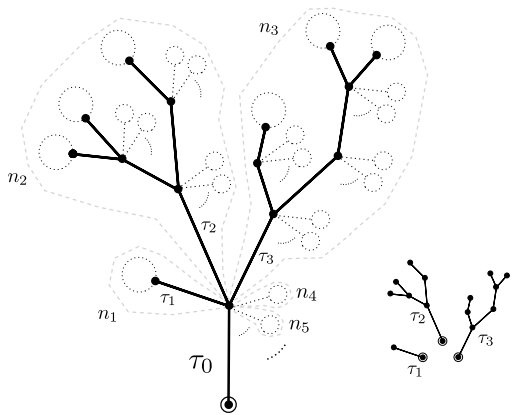}}}
\caption{A diagram explaining (\ref{limgengenstep1}). In this case $k=4$. The dotted circles and lines indicate all the possible trees which have $\tau_0 = \tau_1 \ast \tau_2 \ast \tau_3$ as a left subtree. For $i \leq k-1$, $n_i$ denotes the total volume of the tree which contains $\tau_i$ as a subtree and for $i\geq k$, $n_i$ denotes the total volume of each of the other trees which are attached to the nearest neighbour of the root.} \label{f:measure}
\end{figure}
Therefore we will now only consider one term in the sum over $\ell$ 
\begin {eqnarray} \nonumber
&& \sum_{\substack{n_1 + \cdots + n_{\ell-1} \\ = n-1}}q_\ell(n_1,\ldots,n_{\ell-1})  \prod_{i=1}^{\ell-1} \nu_{n_i}(\tau_i) \\ \nonumber
&& =  \frac{\Gamma\left(\ell-2+\frac{1-\alpha}{a}\right)\G{1-a}\G{n}}{\G{\frac{1-\alpha}{a}}\Gamma(\ell)a \G{n-a}} \sum_{\substack{n_1 + \cdots + n_{\ell-1}  = n-1 \\ n_i \geq 1, ~\forall i}} \prod_{i=1}^{\ell-1} \frac{a \G{n_i-a}}{\G{1-a}\G{n_i+1}} \nu_{n_i}(\tau_i)\\  \label{limgengenstep2}
&& \times ~\left(1-a - \alpha + \alpha \sum_{i=1}^{\ell-1}\frac{n_i}{n-n_i}\right). \label{stest} \nonumber \\
\end {eqnarray}
There is always some index $J$ such that $n_J \geq (n-1)/(\ell-1)$ in the above sum. First, consider the contribution from the constant term $1-a-\alpha$ inside the parentheses in (\ref{stest}), call it $\Sigma_1(n)$.
 Using (\ref{ordergamma})  we find the upper bound
\begin {eqnarray*}
 \Sigma_1(n) &\leq& C_1(\ell) n^a \left(\frac{n-1}{\ell-1}\right)^{-a-1} \sum_{m=\ell-2}^{n-2} \sum_{\substack{n_2+\cdots+n_{\ell-1} = m \\ n_i \geq 1, ~\forall i}}\prod_{j=2}^{\ell-1}\frac{a \G{n_j-a}}{\G{1-a}\G{n_j+1}} 
\end {eqnarray*}
where $C_1(\ell)> 0$ is independent of $n$ and we renamed $m = n-1-n_J$. Using (\ref{aandb}), (\ref{sumofgamma}) and a generating function argument as in Lemma \ref{l:explicit} we find that the summand in the sum over $m$ is $O(m^{-\alpha})$. Therefore $\Sigma_1(n) = O(n^{-\alpha})$ which goes to zero as $n \rightarrow \infty$. 

Next, consider the contribution from terms in the sum inside the  parentheses in (\ref{stest}) for which $i \neq J$, call it $\Sigma_2(n)$. It can be estimated from the above by
\begin {equation}
 \Sigma_2(n) \leq C_2(\ell) n^a \left(\frac{n-1}{\ell-1}\right)^{-a-2} \sum_{m=\ell-2}^{n-2} \sum_{\substack{n_2+\cdots+n_{\ell-1} = m \\ n_i \geq 1, ~\forall i}}n_2\prod_{j=2}^{\ell-1}\frac{a \G{n_j-a}}{\G{1-a}\G{n_j+1}} 
\end {equation}
where $C_2(\ell)> 0$ is independent of $n$. As before, we find that the summand in the sum over $m$ is $O(m^{-\alpha+1})$ and thus $\Sigma_2(n) = O(n^{-\alpha})$ which goes to zero as $n \rightarrow \infty$.

Finally, consider the contribution to (\ref{stest}) from the term in the sum inside the  parentheses in (\ref{stest}) for which $i = J$ and for which $\sum_{j\neq J} n_j > A$ where $A>\ell-2$. Call this contribution $\Sigma_3(n,A)$. We find that
\begin {equation}
 \Sigma_3(n,A) \leq C_3(\ell) n^a \left(\frac{n-1}{\ell-1}\right)^{-a}  \sum_{m > A} \frac{1}{m+1} \sum_{\substack{n_2+\cdots+n_{\ell-1} = m \\ n_i \geq 1, ~\forall i}}\prod_{j=2}^{\ell-1}\frac{a \G{n_j-a}}{\G{1-a}\G{n_j+1}} \label{sigma3est}
\end {equation}
where $C_3(\ell)> 0$ is independent of $n$. It is evident from (\ref{aandb}) and (\ref{sumofgamma}) that the sum over $m$ is convergent. Therefore, $\limsup_{ n\rightarrow\infty} \Sigma_3(n,A) \rightarrow 0$ as $A \rightarrow \infty$. Thus,
\begin {eqnarray} \nonumber
&& \limsup_{n\rightarrow\infty} \sum_{\substack{n_1 + \cdots + n_{\ell-1} \\ = n-1}}q_\ell(n_1,\ldots,n_{\ell-1})  \prod_{j=1}^{\ell-1} \nu_{n_j}(\tau_j) = \limsup_{n\rightarrow \infty} \Sigma_3(n,A) \\ \nonumber
&& +~\frac{\alpha \Gamma\left(\ell-2+\frac{1-\alpha}{a}\right)}{\G{\frac{1-\alpha}{a}}\Gamma(\ell)}\sum_{J=1}^{\ell-1} \nu(\tau_J) \sum_{m = \ell - 2}^{A} \frac{1}{m+1} \sum_{\substack{n_2+\cdots+n_{\ell-1} = m \\ n_i \geq 1, ~i \neq J,~ n_J = 0}}~\prod_{\substack{j=1 \\ j \neq J}}^{\ell-1}\frac{a \G{n_j-a}}{\G{1-a}\G{n_j+1}}\nu_{n_j}(\tau_j) \\ \nonumber
&& \xrightarrow[A \rightarrow \infty] {} \frac{\alpha \Gamma\left(\ell-2+\frac{1-\alpha}{a}\right)}{\G{\frac{1-\alpha}{a}}\Gamma(\ell)}\sum_{J=1}^{\ell-1} \nu(\tau_J) \sum_{m = \ell - 2}^{\infty} \frac{1}{m+1} \sum_{\substack{n_2+\cdots+n_{\ell-1} = m \\ n_i \geq 1, ~i \neq J,~ n_J = 0}}~\prod_{\substack{j=1 \\ j \neq J}}^{\ell-1}\frac{a \G{n_j-a}}{\G{1-a}\G{n_j+1}}\nu_{n_j}(\tau_j). \\ \label{convlimit}
\end {eqnarray}
It is clear the the same holds for the $\liminf$ and therefore (\ref{convseq}) converges for $\tau_0$. Furthermore, the last expression (\ref{convlimit}) along with the estimate (\ref{sigma3est}) on $\Sigma_3$ characterizes the measure $\nu$. The distribution of the vertex degrees on the spine also follows from (\ref{limitqgengen}).
\end {proof}
\begin {proof} [Proof of Lemma \ref{l:annhaus}] \label{pr:annhaus}
 The proof involves deriving and solving a differential equation for certain  generating functions. We start by defining a generating function for the quantity of interest
\begin {eqnarray}
H(x) &=& \sum_{R=1}^\infty \Big\langle \sum_i |B_R(\tau_i)| \Big\rangle_{\mu} x^R,\\
\end {eqnarray}
and we define the probability generating function for the volume of a graph ball of radius $R$ which is centered on a vertex on the spine and intersects the outgrowths from that vertex
\begin {eqnarray}
f_R(z) &=& \sum_{j=0}^\infty \mu\Big(\Big\{ (\tau_1,\tau_2,\ldots) : \sum_{i} |B_R(\tau_i)| = j\Big\}\Big) z^j. \\
\end {eqnarray}
It is furthermore convenient to define the following rescaled probability generating function, with respect to $\nu_n$, of the volume of a ball of radius $R$
\begin {eqnarray}
A_{R,n}(z) &=& \frac{\G{n-a}}{\G{1-a}\G{n+1}} \sum_{j=0}^\infty \nu_n\Big(\Big\{ \tau : |B_R(\tau)| = j\Big\}\Big) z^j, \\
\end {eqnarray}
and the corresponding generating function over $n$
\begin {eqnarray}
G_R(z,\zeta) &=& \sum_{n=1}^\infty A_{R,n}(z) \zeta^n. \\
\end {eqnarray}
We furthermore define
\begin {eqnarray} 
V_R(\zeta) &=& \partial_z G_R(z,\zeta)|_{z=1} \quad \text{and} \\
Q_x(\zeta) &=& \sum_{R=1}^\infty V_R(\zeta) x^R.
\end {eqnarray}
It is straightforward to show that
\begin {equation} \label{GR1}
G_R(1,\zeta) = \frac{1-(1-\zeta)^a}{a}
\end {equation}
for all $R\geq 1$. Using (\ref{phiinfgen}), (\ref{muinfgen}) and (\ref{mudef}) we find the relation
\begin {equation}
 f_R(z) = \alpha \int_{0}^1 \left(1-a G_R(z,\zeta)\right)^{\frac{\alpha-1}{a}} d\zeta
\end {equation}
and therefore, it follows from (\ref{GR1}) that
\begin {equation} \label{brinf}
\Big\langle \sum_i |B_R(\tau_i)| \Big\rangle_{\mu} = f'_R(1) = \alpha (1-\alpha) \int_0^1 V_R(\zeta') (1-\zeta')^{\alpha - a - 1} d\zeta'
\end {equation}
and
\begin {equation} \label{Hfunc}
 H(x) = \alpha (1-\alpha) \int_0^1 Q_x(\zeta) (1-\zeta')^{\alpha - a - 1} d\zeta'.
\end {equation}
Thus, knowing $Q_x(\zeta)$, we can expand $H(x)$ and find its coefficients, which completes the proof. We proceed as follows: Using the Markov branching property we can derive the following recursion
\begin {equation}
 n A_{R,n} =  z \sum_{k=2}^\infty \frac{\G{k-2+\frac{1-\alpha}{a}} a^{k-2}}{\G{\frac{1-\alpha}{a}} \G{k}} \sum_{\substack{n_1 + \cdots n_{k-1} = n-1 \\ n_i \geq 1, ~\forall i}}\left(1-a-\alpha + \alpha \sum_{j=1}^{k-1}\frac{n_j}{n-n_j}\right) \prod_{i=1}^{k-1} A_{R-1,n_i}(z) \label{Arec}
\end {equation}
for $R,n \geq 2$.  In terms of the generating function $G_R(z,\zeta)$, (\ref{Arec}) can be written as
\begin {equation}
 G_1(z,\zeta) = \frac{z\big(1-(1-\zeta)^a\big)}{a}
\end {equation}
for $R=1$ and
\begin {equation}
 \partial_\zeta G_R(z,\zeta) = z\left(\left(1-a G_{R-1}(z,\zeta)\right)^{\frac{a-1+\alpha}{a}}+\alpha \partial_\zeta G_{R-1}(z,\zeta)\int_0^{\zeta} \left(1-a G_{R-1}(z,\zeta')\right)^{\frac{\alpha-1}{a}}d\zeta'\right) 
\end {equation}
for $R\geq 2$. By differentiating the above with respect to $z$, rearranging and differentiating again with respect to $\zeta$ we get the following recursion
\begin {equation}
 V_1(\zeta) = \frac{1-(1-\zeta)^a}{a}
\end {equation}
and
\begin {eqnarray*}
&& (1-\zeta)V_R''(\zeta) - (1-a)V_R'(\zeta) = (1-\zeta)\big(1-(1-\zeta)^\alpha\big) V_{R-1}''(\zeta) \\
&& \qquad \qquad +~ (1-a)(2(1-\zeta)^\alpha-1)V_{R-1}'(\zeta)+a(1-a)(1-\zeta)^{\alpha-1}V_{R-1}(\zeta).
\end {eqnarray*}
Writing this in terms of the function $Q_x(\zeta)$ we find the differential equation
\begin {eqnarray*}
 &&(1-\zeta)\left(1-\frac{1}{x} - (1-\zeta)^\alpha\right) Q''_x(\zeta) + (1-a)\left(2(1-\zeta)^\alpha - 1 + \frac{1}{x}\right) Q'_x(\zeta) \\
&& +a(1-a)(1-\zeta)^{\alpha-1} Q_x(\zeta)= 0
\end {eqnarray*}
with initial conditions
\begin {equation}
 Q_x(0) = 0 \qquad \text{and} \qquad Q'_x(0) = \frac{x}{1-x}.
\end {equation}
We will solve this differential equation with straightforward methods and by plugging the solution into (\ref{Hfunc}) we can determine the function $H(x)$ and extract the coefficients of its power series.  We start by making a convenient change of variables by defining $y(\zeta) = (1-\zeta)^\alpha (1-x^{-1})^{-1}$ and  
\begin {equation} \label{Pdef}
P_x(y(\zeta)) = Q_x(\zeta)-\frac{y(\zeta)}{\alpha} + \frac{x}{\alpha(x-1)}. 
\end {equation}
Then $P_x(y)$ satisfies the following inhomogeneous, hypergeometric differential equation
\begin {eqnarray} \nonumber
&& \alpha y (1-y) P_x''(y) + (\alpha - a - (\alpha-2 a +1)y) P_x'(y) + \frac{a(1-a)}{\alpha}P_x(y) \\
&& = \frac{(\alpha-a+1)(\alpha-a)}{\alpha^2}y -\frac{\alpha-a}{\alpha} + \frac{a(1-a)x}{\alpha^2 (x-1)} \label{Pinfgen}
\end {eqnarray}
with initial conditions
\begin {equation} \label{boundinfgen}
 P_x\left(\frac{x}{x-1}\right) = P'_x\left(\frac{x}{x-1}\right) = 0.
\end {equation}
A basis of solutions to the homogeneous equation is 
\begin {equation}
 u_1(y) = (-y)^{\frac{a}{\alpha}} \quad \text{and} \quad u_2(y) = (1-y)^{-\frac{1-a}{\alpha}} {}_2F_1\left(\frac{1-a}{\alpha},1,\frac{\alpha-a}{\alpha};\frac{y}{y-1}\right).
\end {equation}
The Wronskian is given by
\begin{eqnarray*} \nonumber
W(y) &=& \left|\begin {array} {cc}    
  u_1(y) & u_2(y) \\
  u_1'(y)& u_2'(y) \\
 	\end {array}\right| 
=  \frac{a}{\alpha} (-y)^{\frac{a-\alpha}{\alpha}}(1-y)^{-\frac{1-a}{\alpha}}. 
\end{eqnarray*}
The Green's function corresponding to (\ref{Pinfgen}) is defined as
\begin{eqnarray*}
 F(y,v) &=& \frac{1}{\alpha v(1-v) W(v)}\bigg(-u_2(v)u_1(y) + u_1(v)u_2(y)\bigg). 
\end{eqnarray*}
The solution to (\ref{Pinfgen}) and (\ref{boundinfgen}) is then
\begin {equation} \label{solutiongenP}
 P_x(y) = \int_{\frac{x}{x-1}}^y F(y,v) \left( \frac{(\alpha - a +1)(\alpha-a)}{\alpha^2}v -\frac{\alpha-a}{\alpha} + \frac{a(1-a)x}{\alpha^2 (x-1)} \right) dv.
\end {equation}
Using (\ref{Hfunc}) and (\ref{Pdef}) we find that
\begin {eqnarray} \nonumber
H(x) &=& (1-\alpha) \left(\frac{1-x}{x}\right)^{\frac{\alpha-a}{\alpha}} \int_{\frac{x}{x-1}}^0 P_x(y) dy \\  \label{relationBRP}
&& +~\frac{\alpha(1-\alpha)x}{(\alpha-a)(2\alpha-a)(1-x)}.
\end {eqnarray}
Plugging the solution (\ref{solutiongenP}) into (\ref{relationBRP}) and using the formula
\begin {equation}
 \int_a^t \int_a ^x f(s) ds dx = \int_a ^t f(s) (t-s) ds
\end {equation}
we obtain, with some rewriting,
\begin {eqnarray} \nonumber
H(x) &=& (1-\alpha) \alpha^{-2}x^{\frac{a-\alpha}{\alpha}} (1-x)^{-\frac{a}{\alpha}} \int_{1-x}^1 {}_2F_1\left(\frac{1-a}{\alpha},1,\frac{\alpha+1}{\alpha};v\right) v^{\frac{a-3\alpha}{\alpha}}(1-v)^{-\frac{a}{\alpha}} \\ \nonumber
&& \times~ \Bigg[\alpha^2 x(1-x)^{\frac{a-1}{\alpha}} v^{\frac{1-a+\alpha}{\alpha}} +(1-a)(a-\alpha + \alpha x) v^{2}\\ \nonumber
&& +~ (1+\alpha-a)(a x - 2\alpha x -2 a + 2\alpha) v + (\alpha - a) (a-1-2\alpha)(1-x)\Bigg] dv \\  \label{summan}
&& +~ \frac{\alpha(1-\alpha)x}{(\alpha-a)(2\alpha-a)(1-x)}. 
\end {eqnarray}
It is possible to simplify this even further. By considering solutions for a few choices of convenient parameters (e.g.~$D=3$ or $D=\infty$, $\gamma = 1$) one immediately guesses that
\begin {equation} \label{nicesolution}
 H(x) = \frac{1-\alpha}{1-a-\alpha} (1-x)^{-1}\left({}_2 F_1\left(1,\frac{1-\alpha-a}{\alpha},\frac{\alpha-a}{\alpha};x\right)-1\right).
\end {equation}
It is then straightforward to check that this is true by rearranging and differentiating (\ref{summan}).
The result follows by expanding (\ref{nicesolution}).
\end {proof}
\begin {proof} [Proof of Theorem \ref{th:ashaus}] \label{pr:ashaus}
Let $\lambda(R)$ be a positive function with the property that
\begin {equation}
\sum_{R=1}^\infty \frac{1}{R\lambda(R)} < \infty.
\end {equation}
We will show that there exist constants $C_1$ and $C_2$ and for $\nu$--almost all trees $T$ a constant $R_T > 0$ such that
\begin {equation} \label{ieh}
 C_1 (\log(R)^{-1}R)^{1/\alpha} \leq |B_R(T)| \leq C_2(\lambda(R) R)^{1/\alpha}
\end {equation}
for all $R \geq R_T$. 
In particular we can choose $\lambda(R) = (\log(R))^\eta$ for any $\eta>1$ which is sufficient to obtain the desired result.

On an infinite, rooted tree with a single spine, denote the vertex on the spine at distance $n$ from the root by $s_n$, $n\geq1$. Let $(X_n(T))_n$ be a sequence of random variables corresponding to the number of leaves attached to the vertices $s_1,s_2,\ldots$ of $T$. Define $S_R(T) = \sum_{i=1}^R X_i(T)$. Then $|B_R(T)| = S_{R-1}(T) + R$ and it is clearly sufficient to prove the above inequalities for $S_R(T)$. Begin with the lower bound. Take $\kappa, \theta > 0$. Using Markov's inequality, the independence of the $X_i$'s and Equation (\ref{phiinfcat}) we get
\begin {eqnarray*}
\nu(S_R(T) < \kappa) &=& \nu\left(e^{-\theta S_R} > e^{-\theta\kappa}\right) \leq e^{\theta\kappa} \left(\langle e^{-\theta X_i(T)}\rangle_{\nu}\right)^R \\
&=& e^{\theta\kappa} \left(1-\left(1-e^{-\theta}\right)^\alpha\right)^R \\
&\leq& e^{\theta\kappa - R\left(1-e^{-\theta}\right)^\alpha}.
\end {eqnarray*}
Now choose $\kappa = K (\log(R))^{-1/\alpha}R^{1/\alpha}$ and $\theta = 1/\kappa$. Note that
\begin {equation}
(1-e^{-\theta})^\alpha = \theta^\alpha(1 + O(\theta)) 
\end {equation}
and therefore for $R$ large enough
\begin {equation}
\nu(S_R(T) < K (\log(R))^{-1/\gamma}R^{1/\alpha}) \leq C_3 e^{-K^{-\alpha}\log(R)} = C_3 R^{-K^{-\alpha}}
\end {equation}
where $C_3$ is a positive constant. Choosing $K = C_1$ small enough we see that 
\begin {equation}
 \sum_{R=1}^\infty  \nu(S_R(T) < C_1 (\log(R))^{-1/\gamma}R^{1/\alpha}) < \infty
\end {equation}
and therefore, by the Borel--Cantelli lemma, there exists a constant $R_T$ such that $S_R(T) \geq C_1 (\log(R))^{-1/\alpha}R^{1/\alpha}$ almost surely for all $R\geq R_T$.

The upper bound follows from Theorem 2 in \cite{feller} which states, for our purposes, that the probability of the event
\begin {equation}
 S_R(T) > a_R, \quad \quad \text{for infinitely many $R$}
\end {equation}
is zero if the sum
\begin {equation}
 \sum_{R=0}^\infty \nu(X_k \geq a_R)
\end {equation}
converges, where $a_R$ is a positive sequence with the property that $a_R/R \rightarrow \infty$ as $R\rightarrow \infty$. According to (\ref{phiinfcat})
\begin {equation}
 \nu(X_k \geq a_R) \leq C_4 a_R^{-\alpha}
\end {equation}
for a suitable constant $C_4$. Choosing $a_R = (\lambda(R)R)^{1/\alpha}$, where $\lambda(R)$ has the properties stated above, completes the proof.
\end {proof} 

\begin {proof} [Proof of Theorem \ref{th:spec}] \label{pr:spec}
For a tree $T$, define the generating function
\begin {equation}
 Q_T(x) = \sum_{t=0}^\infty p_T(t) (1-x)^{\frac{t}{2}}.
\end {equation}
From  \cite[Lemma 7 and Equation (6)]{sdgt} we have the following inequalities for any fixed caterpillar $T$ and any integer $R>0$
\begin {equation}
 \frac{1}{\frac{1}{R-1}+x+x|B_R(T)|}\leq Q_T(x) \leq R + \frac{2}{x|B_R(T)|}.
\end {equation}
Using (\ref{ieh}) for a suitable choice of $\lambda(R)$ we get $\nu$--almost surely the inequality
\begin {equation}
\frac{1}{\frac{1}{R-1}+x+x(\lambda(R) R)^{1/\alpha}}\leq Q_T(x) \leq R + \frac{2}{xC_1 (\log(R)^{-1}R)^{1/\alpha} }
\end {equation}
for all $R\geq R_T$ and $R_T$ large enough. Choosing $R=[x^{-\frac{\alpha}{1+\alpha}}]$ we find that there are numbers $K_1(T)$ and $K_2(T)$ such that $\nu$--almost surely
\begin {equation}
K_1(T)\lambda([x^{-\frac{\alpha}{1+\alpha}}])^{-1}x^{-\frac{\alpha}{1+\alpha}}\leq Q_T(x) \leq K_2(T) \log([x^{-\frac{\alpha}{1+\alpha}}])x^{-\frac{\alpha}{1+\alpha}}.
\end {equation}
This yields the desired limit.
\end {proof}

\end{document}